\documentclass[a4paper]{article}
\usepackage{fullpage}
\usepackage{amsmath,amsthm,amssymb}

\usepackage{subfigure}

\usepackage{tikz}
\usetikzlibrary{shapes,arrows}
\usetikzlibrary{matrix}

\usepackage{enumerate}
\usepackage{fullpage}
\usepackage{algorithm,algorithmic}

\newcommand{\INITIALLY}{\REQUIRE{}}
\newcommand{\ROUND}{\ENSURE{}}

\newtheoremstyle{newthm}
  {\topsep}   
  {\topsep}   
  {\itshape}  
  {0pt}       
  {\scshape} 
  {.}         
  {5pt}  
  {}          

\theoremstyle{newthm}

\newtheorem{thm}{Theorem}
\newtheorem{prop}[thm]{Proposition}
\newtheorem{lem}[thm]{Lemma}

\DeclareMathOperator{\dec}{dec}

\newcommand{\IR}{\mathbb{R}}
\newcommand{\IN}{\mathbb{N}}

\renewcommand{\leq}{\leqslant}
\renewcommand{\ge}{\geqslant}
\renewcommand{\geq}{\geqslant}

\usepackage{mathtools}

\DeclareMathOperator\In{In}
\DeclareMathOperator\Out{Out}

\title{Amortized Averaging Algorithms for Approximate Consensus}

\author{Bernadette Charron-Bost\textsuperscript{1} \and Matthias
F\"ugger\textsuperscript{2} \and Thomas Nowak\textsuperscript{3}}
\date{\textsuperscript{1} CNRS, \'Ecole polytechnique\\
\textsuperscript{2} Max-Planck-Institut f\"ur Informatik\\
\textsuperscript{3} Universit\'e Paris-Sud}

\begin{document}
\maketitle

\begin{abstract}

We introduce a new class of distributed algorithms for the approximate consensus problem in  dynamic rooted networks, 
	which we call {\em amortized averaging algorithms}.
They are deduced from ordinary averaging algorithms by adding a value-gathering phase 
	before each value update.
This  allows their decision time to drop from being exponential 
	in the number $n$ of processes to being linear under the assumption that each process knows $n$.
In particular, the {\em amortized midpoint algorithm}, which achieves a linear decision time, works in completely 
	anonymous dynamic rooted networks where processes can exchange and store continuous values, and 
	under the assumption that  the number of processes is known to all processes.
We then study the way amortized averaging algorithms degrade when  communication graphs are from time to time non rooted,
	or with a wrong estimate of the number of processes.
Finally, we analyze the amortized midpoint algorithm under the additional constraint that 
	processes can only store and send quantized values, and get as a corollary that the 2-set consensus problem
	is solvable in linear time in any rooted dynamic network model when allowing all decision values to be in the range of initial values.

\end{abstract}

\setcounter{footnote}{3}

\section{Introduction}

This paper studies the problem of {\em approximate consensus}, i.e.,
	the task for a set of  processes, each of them with an initial real scalar value, 
	to decide on values that lie in an $\varepsilon$-neighborhood of each other 
	and in the range of the initial values.
It has many applications, e.g., for geometric coordination tasks or clock synchronization, 
	which often have to be solved in mobile ad-hoc networks under quite adverse constraints.

In recent work~\cite{CBFN15}, we proved that approximate consensus is solvable in a dynamic network model if and only if all occurring communication graphs are rooted, i.e., contain a rooted spanning tree.
This spanning tree, as well as the root 
	can change completely from one communication graph to the next.
There is hence no stability condition necessary for solving approximate consensus.
In fact, we showed that it suffices to restrict our attention to the simple class of {\em averaging algorithms\/} 
	to solve approximate consensus whenever it is indeed solvable.
In these algorithms, processes have a single scalar state variable, which they repeatedly update to a weighted average of the values they received from their neighbors.
We proved that, if~$\varrho$ is a lower bound on the weights used in the averaging steps of an averaging algorithm, then it achieves $\varepsilon$-agreement in $O\!\left(n\varrho^{-n} \log\frac{1}{\varepsilon}\right)$ rounds where~$n$ is the number of processes.
Moreover, example graphs exist showing that averaging algorithms can really need an exponential time until they achieve $\varepsilon$-agreement.

We now list the main contributions of the present paper:
\begin{enumerate}
\item We propose a new analysis technique of averaging algorithms based on the
	notion of {\em $\varrho$-safeness}, introduced in~\cite{Mor05,Cha11}, which is a generalization of 
	the lower bound condition on the parameter~$\varrho$ discussed above. 
This technique focuses on the essential properties needed for contracting the range of 
	current values in the system by directly focusing on the set of values and not 
	on the weights used in the averaging steps, as done classically.
Together with a graph-theoretic reduction that already played a key role in~\cite{CBFN15}, it enables short proofs of the
	convergence and of upper bounds on the contraction rate of (existing and new) averaging algorithms.

\item We introduce the notion of {\em amortization\/} of averaging algorithms,
	which consists in inserting a value-gathering phase before each averaging step.
This additional phase, surprisingly, transforms averaging algorithms into ``turbo versions'' of 
	themselves and takes their decision time from being exponential~\cite{CBFN15} 
	to being polynomial in the number of processes.
\item We combine these two ideas in the design of the {\em amortized midpoint algorithm}, which we prove to have an optimal contraction rate and whose decision time is linear in the number of processes, which is also optimal.
More specifically, it achieves $\varepsilon$-agreement in $O\!\left(n\log\frac{1}{\varepsilon}\right)$ rounds where~$n$ is the number of processes.
This algorithm neither relies on any stability of the communication topology nor on any way of distinguishing two processes.
It thus works in completely dynamic anonymous networks.
\end{enumerate}

The well-functioning of the amortized midpoint algorithm does, however, rely on two fundamental hypotheses:
	knowing the number of processes and the rootedness of all occurring communication graphs.
However we demonstrate that, even with an erroneous estimate on the number of processes
	or when communication graphs sometimes fail to be rooted, the
amortized midpoint algorithm still achieves $\varepsilon$-agreement, albeit at
a later time.
In fact, this graceful degradation property holds for all amortized averaging
algorithms, which shows that a certain part of the often observed robustness of
averaging algorithms actually carries over to their amortized versions.

The  linear convergence time of the amortized midpoint algorithm, combined with
 	its great versatility and robustness,   is  especially  striking, notably when 
	comparing it to classical averaging algorithms such as the {\em equal-neighbor} 
	algorithm.
It is well known  that the latter averaging algorithms may give an exponential decision time in the number of processes.

One of the first steps to go from exponential to polynomial decision time was done by Olshevsky and Tsitsiklis~\cite{OT11} 
	who presented a cubic-time averaging algorithm with time-varying bidirectional topologies.
This result was later improved to quadratic time~\cite{NOOT09} when the update rules at every time correspond to
	a doubly stochastic matrix.
Very recently, Olshevsky presented a linear-time algorithm in time-constant bidirectional communication graphs~\cite{Ols15}.
This result was preceded by other attempts at lowering the convergence time in several other special cases (e.g., \cite{DCZ09,HJOV14,BHOT05,YSSBG13}).
The sum of these efforts makes the time-linearity of the relatively simple amortized midpoint algorithm in arbitrarily dynamic directed anonymous networks all the more so striking.

The robustness against changes in the hypotheses goes even further:
If processes do not dispose of real-valued variables, but only variables of finite precision, 
	then the amortized midpoint algorithm achieves $\varepsilon$-consensus in linear time 
	as long as it is not obviously impossible, i.e., whenever~$\varepsilon$ is larger than the variables' precision.
If~$\varepsilon$ is smaller than the precision, then $\varepsilon$-agreement cannot be achieved 
	in all rooted network models because of the impossibility of exact consensus~\cite{SW89}.
In fact, this property of achieving $\varepsilon$-agreement whenever it is possible is specific 
	to the amortized midpoint algorithm and does not hold for general amortized averaging algorithms 
	and even less so for non-amortized averaging algorithms.
	
By choosing~$\varepsilon$ equal to the precision, we get as a corollary that the 
	{\em $2$-set consensus problem} is solvable in linear time in all rooted dynamic network models.
This new result on $k$-set consensus holds when all decision values are only constrainted to be in the range 
	of initial values (instead of being equal to initial values), which sheds new light on the role of the validity condition in the
	specification of this classical problem in distributed computing.

While the practical implications of having a fast approximate consensus algorithm are immediately visible in man-made communication networks, our results go even beyond that and are interesting also from a different, more analytic, perspective.
Many natural phenomena, like bird flocking~\cite{Cha09}, synchronization of coupled oscillators~\cite{Str00}, opinion dynamics~\cite{HK02}, or firefly synchronization~\cite{MS90}, can be modeled and explained via agents that execute averaging algorithms.
However, there is an important mismatch between the fast convergence times observed in nature and the theoretical worst-case times, which often are exponential or worse~\cite{Cha09}.
In this sense, our results provide a step towards closing this gap by suggesting that the agents in these natural systems in fact operate on a different, slower, time scale than their environments, and that this inertia actually contributes positively to faster convergence.

The rest of the paper is organized as follows:
Section~\ref{sec:model} gives the problem and model definition, as well as introducing the concept of averaging algorithms.
Section~\ref{sec:nonsplit} introduces the analysis technique we use and gives tight upper bounds on the contracting rate of averaging algorithms.
The concept of amortization of averaging algorithms is given in Section~\ref{sec:amortized}.
Several amortized versions of averaging algorithms are introduced and studied, in particular the amortized midpoint algorithm.
Resiliency and robustness of these algorithms is studied in Section~\ref{sec:resiliency}.
Finally, Section~\ref{sec:rounding} extends the analysis to variables of finite precision.

\section{Approximate consensus and averaging algorithms}\label{sec:model}

We assume a distributed, round-based computational model in the spirit
     of the Heard-Of model by Charron-Bost and Schiper~\cite{CS09}.
A system consists in a set  of processes $[n] = \{1,\dots,n\}$.
Computation proceeds in {\em rounds}:
In a round, each process sends its state to its outgoing neighbors,
	 receives values from its incoming neighbors, and
     finally updates its state based.
The  value of the updated state is determined by a deterministic
     algorithm, i.e., a transition function that maps the values in the
     incoming messages  to a  new state value.
Rounds are communication closed in the sense that no process receives
     values in round~$k$ that are sent in a round different from~$k$.

Communications that occur in a round are modeled by a directed graph~$G=([n], E(G))$ 
	with  a self-loop at each node.
The latter requirement is quite natural as a process can obviously communicate with 
	itself instantaneously.
Such a directed graph is called a  {\em communication graph}.
We denote by~$\In_p(G)$ the set
     of incoming neighbors of~$p$ and by $\Out_p(G)$ the set of
     outgoing neighbors of~$p$ in~$G$.
     
The {\em product\/} of two communication graphs $G$ and $H$, denoted $G \circ H$, is 
	the directed graph  with an edge from $(p,q)$ 
	 if there exists~$r\in [n]$ such that $(p,r)\in E(G)$ and $ (r,q) \in E(H)$.
We easily see that $G  \circ H$ is a communication graph and, because of the self-loops, 
	$E(G) \cup E(H) \subseteq E(G  \circ H) $.

A {\em communication pattern\/} is a sequence~$(G(k))_{k\ge 1}$ of
     communication graphs.
For a given communication pattern,
	$\In_p(k)$ and  $\Out_p(k)$ stand for  $\In_p(G(k))$ and $\Out_p(G(k))$, respectively.

Each process~$p$ has a {\em local state\/} $s_p$ the values of which at
     the end of round~$k \ge 1$ is denoted by~$s_p(k)$.
Process~$p$'s initial state, i.e., its state at the beginning of round~$1$,
     is denoted by~$s_p(0)$.
Let the {\em global state} at the end of round~$k$ be the collection
     $s(k) = (s_p(k))_{p \in [n]}$.
The {\em execution\/} of an algorithm from global initial state~$s(0)$,
     with communication pattern~$(G(k))_{k\ge 1}$ is the unique
      sequence $(s(k))_{k\ge 0}$ of global states defined as follows:  for  each round~$k \ge 1$, 
      process~$p$ sends~$s_p(k-1)$ to all the processes in~$\Out_p(k)$,
     receives $s_q(k-1)$ from each process~$q$ in~$\In_p(k)$, and
     computes~$s_p(k)$ from the incoming messages, according to the
     algorithm's transition function.
     
When  the structure of states allows each process to record and to relay information it has received during
	any period of $K$ rounds for some positive integer~$K$,   we may be led to  modify time-scale and to
	consider blocks of $K$  consecutive rounds, called {\em macro-rounds}:
	macro-round $\ell$ is the sequence of rounds $(\ell-1) K+1, \ldots, 	\ell K$ and the corresponding
	information  flow graph, called {\em communication graph at macro-round} $\ell$,  is 
	the product of the communication graphs $ G((\ell-1)K +1)  \circ \ldots \circ G( \ell K ) $.

\subsection{Consensus and approximate consensus}
 
A crucial problem in distributed systems is to achieve agreement among local process states
     from arbitrary initial local states.
It is a well-known fact that this goal is not easily achievable in the context 
	of dynamic network changes~\cite{SW89,CBFN15}, and restrictions on communication patterns
	are required for that.
We define a {\em network model\/} as a non-empty set ${\cal N}$ of
     communication graphs, those that may occur in communication patterns.

We now consider the above round-based algorithms in which the local state of 
	process~$p$ contains two variables~$x_p$ and~$\dec_p$. 
Initially the range of $x_p$ is $[0,1]$
 	and $\dec_p = \bot$ (which informally means that $p$
	has not  decided yet).\footnote{%
	In the case of  {\em binary consensus\/}, $x_p$  is restricted to be initially from $\{0,1\}$.}
Process~$p$ is allowed to set~$\dec_p$ to the current value of $x_p$, and so to a value~$v$ different from $\bot$,
	     only once; in that case we say that~$p$ {\em decides}~$v$.

For any  $\varepsilon \geq 0$, an algorithm {\em  achieves $\varepsilon$-agreement 
	with the communication pattern}~$(G(k))_{k\ge 1}$  if each execution from a global initial state as 
	specified above and with the communication pattern~$(G(k))_{k\ge 1}$   fulfills the following three conditions:               
	\begin{description}
	  \item{\em $\varepsilon$-Agreement.} The decision values of any two processes are within~$\varepsilon $.

	  \item{\em Validity.}  All decided values are in the range of the initial values of processes.

	  \item{\em Termination.} All processes eventually decide.

	\end{description}
An algorithm {\em  solves approximate consensus\/} in a network
	     model~${\cal N}$ if for any $\varepsilon >0$, it achieves $\varepsilon$-agreement with each communication
	     pattern formed with graphs all in~${\cal N}\!.$

In a previous paper~\cite{CBFN15}, we proved the following characterization of network models in which 
	approximate consensus is solvable:

\begin{thm}[\cite{CBFN15}]\label{thm:CBFN15}
The approximate consensus problem is solvable in a network model ${\cal N}$ 
	if and only if each graph in ${\cal N}$ has a rooted spanning tree.
\end{thm}

\subsection{Averaging algorithms }

We focus on  {\em averaging algorithms} in which at each round~$k$, 
	every process~$p$ updates~$x_p$ to some weighted average
	of the values it has just received: 
	formally, if  $m_p(k-1)$  is the minimum of the values $\{ x_q(k-1) \mid q\in \In_p(k)\}$ 
	received by~$p$ in round~$k$ and $M_p(k-1)$ is its maximum, then
	$$  m_p(k-1) \leq x_p(k) \leq   M_p(k-1)  \, .$$
In other words,  at each round~$k$, every process adopts a new value within the interval
 	formed by the values of its incoming neighbors in the communication graph~$G(k)$.

Since we strive for distributed implementations of averaging
     algorithms, weights in  the average update rule of  process~$p$ are required 
     to be locally computable by~$p$.
They may depend only on the set of values received by~$p$ at round~$k$, as is the
	case, for instance, with the update rule of the {\em mean-value algorithm\/}:
	\begin{equation}\label{eq:MV}
	  x_{p}(k) = \frac{1}{|V_p(k)|} \sum_{v \in V_p(k)}  v    \, 
	  \end{equation}
	where $ V_p(k) = \{  x_q(k-1) \mid q\in \In_p(k) \}$.  
In contrast, even in anonymous networks,  weights in the update rule for~$p$  may depend on the 
	{\em multiset} of values received by~$p$ at round~$k$ counted with their multiplicities:
	as an example, $p$'s update rule in the
	 {\em equal-neighbor algorithm\/}  is given by:
	\begin{equation}\label{eq:EN}
	x_{p}(k) = \frac{1}{|\In_p(k)|} \sum_{q\in \In_p(k)}  x_q(k-1) \, .
	\end{equation}
     
Observe that the decision rule is not specified in the above definition of averaging algorithms:
	the decision time immediately follows from the number of rounds that is proven 
	to be sufficient to reach $\varepsilon$-agreement.

Let $\varrho\in ]0,1/2]$;
  	an averaging algorithm is {\em $\varrho $-safe in}   ${\cal N}$  if at any round~$k$ of each of its executions
	with communication patterns in  ${\cal N}$,  each process adopts a new value within the interval formed
	by its neighbors in~$G(k)$  not too close to the boundary: 
	$$ \varrho M_p(k-1) + (1-\varrho) m_p(k-1) \leq x_p(k) \leq  (1-\varrho) M_p(k-1) + \varrho m_p(k-1) \, .$$
Clearly if an averaging algorithm is $\varrho $-safe, then it is $\varrho ' $-safe for any 
	$\varrho ' \leq \varrho$.
We also easily check the following property of the equal-neighbor  and mean-value algorithms. 
 
\begin{prop}\label{prop:EN+MV}
In any network model with $n$ processes, the equal-neighbor algorithm and the mean-value algorithm 
	are both $(1/n)$-safe.
\end{prop}

Finally, an averaging algorithm is $\alpha$-{\em contracting in } ${\cal N}$ if
	at each round~$k$ of each of its executions with communication patterns in  ${\cal N}$,  
	we have
	$$\delta \big( x(k) \big) \leq \alpha \, \delta \big( x(k-1) \big)  $$
	where  $\delta$ is the seminorm on $\IR^n$ defined by
	$\delta(x)= \max_{p} (x_p) - \min_{p} (x_p) $. 

\begin{prop}\label{prop:contract}
In any network model, every $\alpha$-contracting averaging algorithm with $\alpha \in [0,1[$ 
	solves approximate consensus and  achieves $\varepsilon$-agreement in  
	$\big\lceil \log_{\frac{1}{\alpha}}  \big( \frac{1}{\varepsilon} \big) \big\rceil$ rounds.
\end{prop}	

\begin{proof}
Let $m(k)$ and $M(k)$ denote the minimum value and the maximum value in the network 
	at round~$k$, respectively; hence $\delta \big(x(k) \big)=M(k)-m(k) $.
	
An easy inductive argument shows that the sequences $\big( m(k) \big)_{k \geq 1}$ and 
	$\big( M(k) \big)_{k \geq 1}$ are non-decreasing and non-increasing, respectively.
Because the averaging algorithm is $\alpha$-contracting, we have 
	$$\delta \big( x(k) \big) \leq  \delta \big(x(0)\big) \alpha^k \, . $$
 Since $\alpha <1$, the two sequences  $\big( m(k) \big)_{k \geq 0}$ and  $\big( M(k) \big)_{k \geq 0}$ converge to the same  value.
The convergence proof of each sequence $\big( x_p(k) \big)_{k \geq 1}$ to a common value in 
	the range of the initial values rests on the inequalities 
	$$ m(0) \leq m(k) \leq x_p(k) \leq M(k) \leq M(0) \, .$$
We have $\delta \big(x(k)\big) \leq \alpha^{k}$ 
	for every integer $k\geq 0$.
Plugging in $K = \big\lceil \log_{\frac{1}{\alpha}}  \big( \frac{1}{\varepsilon} \big) \big\rceil $ gives
	$\delta \big(x(K)\big) \leq  \varepsilon$.
\end{proof}

\section{Averaging algorithms, nonsplitness, and contraction rates}\label{sec:nonsplit} 

In a previous paper~\cite{CBFN15}, we proved solvability of approximate consensus in 
	rooted network models by a reduction to {\em nonsplit} network models:
	a directed graph is {\em nonsplit\/}  if any two nodes have a common 
	incoming neighbor. 
Indeed we showed the following general proposition:

\begin{prop}[{\cite{CBFN15}}]\label{prop:productrooted}
Every product of $n-1$ rooted graphs  with $n$ nodes and self-loops at all nodes is nonsplit.
\end{prop}

\noindent The main point of  nonsplitness then lies in the following	 result:
\begin{thm}\label{thm:nonsplit}
In a nonsplit network model, a $\varrho$-safe averaging algorithm is $(1-\varrho)$-contracting.
Thus it solves approximate consensus and achieves $\varepsilon$-agreement in  
	$\big\lceil \log_{\frac{1}{1-\varrho}}  \big( \frac{1}{\varepsilon} \big) \big\rceil$ rounds.
\end{thm}

\begin{proof}
Since the communication graph~$G(k)$ is nonsplit, any two processes $p$ and $q$ receive 
	at least one common value~$v$.
It follows that:
	$$ \varrho  v + (1-\varrho) m_p(k-1) \leq x_p(k) \leq  (1-\varrho) M_p(k-1) + \varrho v $$
and
$$	 \varrho  v + (1-\varrho) m_q(k-1) \leq x_q(k) \leq  (1-\varrho) M_q(k-1) + \varrho v  \, .$$
Hence $$	\big \lvert x_p(k) - x_q(k)  \big \rvert \leq (1-\varrho)\cdot \max_{p,q} \big( M_p(k-1)- m_q(k-1) \big) = (1-\varrho)\cdot\delta\big(x(k-1)\big) \, . $$
This shows that the algorithm is $ (1-\varrho) $-contracting, and the rest of the theorem follows from 
	Proposition~\ref{prop:contract}.
\end{proof}

Combined with Proposition~\ref{prop:EN+MV}, we obtain an elementary proof of the classical result~\cite{CMA08a,Cha13}
	that approximate consensus is solvable in rooted network models, 
	which does not make use anymore of Dobrushin's coefficients of scrambling matrices.

As an immediate consequence of Theorem~\ref{thm:nonsplit}, we deduce that any $\varrho$-safe algorithm 
	is at best $(1/2)$-contracting since by definition the safety coefficient~$\varrho$ cannot be larger than~$1/2$.
Indeed, in Section~\ref{sec:amortized} we shall present an averaging algorithm that is  $(1/2)$-safe in any 
	nonsplit network model.
But do there exist other averaging algorithms with a contraction rate less than $1/2$?
In the rest of the section, we show that the answer is no in the network model of nonsplit communication graphs,
	 except in the case of two processes for which
	we prove that the best contraction rate is $1/3$.

\begin{thm}\label{thm:contractrate}
In the  network model of nonsplit communication graphs with $n$ processes, there is no 
	averaging algorithm
        that is  $\alpha$-contracting  
	for any $\alpha > 1/3$ if $n=2$ and for any $\alpha > 1/2$ if $n \geq 3$.
\end{thm}

\begin{proof}
Let ${\cal A}$ be an averaging algorithm that solves  approximate consensus and that is 
	$\alpha$-contracting.

If $G$ and $H$ are two communication graphs with the same set of nodes $[n]$,
	we say that $G$ and $H$ are {\em equivalent with respect to} $p \in [n]$,
	denoted $G\sim_p H$, when $\In_p(G) = \In_p(H)$.
	
We first study the case $n=2$ and we consider the three rounds starting with  $x_1(0)=0$, $x_2(0)=1$,
	and the three communication graphs~$G$, $H^+$, and $H^-$  (see Figure~\ref{fig:thm:lower:bound:markov:graphs}).

Because of the definition of an averaging algorithm, with the communication graph~$H^+$ we have 
	$$x_1(1) = x_1(0)=0  \    \mbox{ and } \     x_2(1) = b \in [ 0, 1] \, .$$
Likewise, with~$H^-$, we have 
	$$x_2(1) = x_2(0)=1 \    \mbox{ and } \   x_1(1) = a \in [ 0, 1]   \, . $$
Since $G \sim_2 H^+$ and $G \sim_1 H^-$,  the vector $x(1)$ obtained with the communication graph~$G$
	satisfies $$x_1(1) = a \    \mbox{ and } \     x_2(1) = b \, .$$
Because ${\cal A}$ is $\alpha$-contracting by assumption, we get
	$$  b  \leq \alpha, \quad    1-a  \leq \alpha,\ \mbox{ and } \   \lvert a-b\rvert \leq \alpha \, .$$
Summing these three inequalities gives $\alpha \geq 1/3$ for $n=2$ as required.

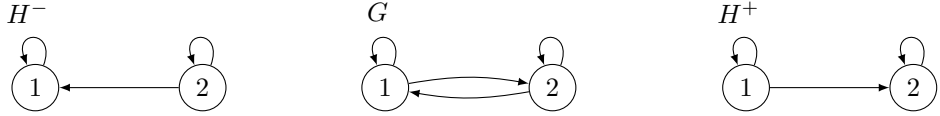
\begin{figure}
\centering
\begin{tikzpicture}[>=latex]
\tikzset{every loop/.style={min distance=5mm,in=0,out=60,looseness=5}}

\node at (-1.2,1.0) {${H^-}$};

\node[draw, circle] (n1) at (180:1.1) {$1$};
\node[draw, circle] (n2) at (0:1.1) {$2$};

\draw[<-] (n1) edge    (n2);
\path[->] (n1) edge [in=110,out=70,looseness=6] (n1);
\path[->] (n2) edge [in=110,out=70,looseness=6] (n2);



\end{tikzpicture}
\hspace{1.5cm}
\begin{tikzpicture}[>=latex]
\tikzset{every loop/.style={min distance=5mm,in=0,out=60,looseness=5}}

\node at (-1.2,1.0) {${G}$};

\node[draw, circle] (n1) at (180:1.1) {$1$};
\node[draw, circle] (n2) at (0:1.1) {$2$};


\draw[->] (n1) edge   [bend left=10] (n2);
\draw[<-] (n1) edge   [bend right=10] (n2);
\path[->] (n1) edge [in=110,out=70,looseness=6] (n1);
\path[->] (n2) edge [in=110,out=70,looseness=6] (n2);



\end{tikzpicture}
\hspace{1.5cm}
\begin{tikzpicture}[>=latex]
\tikzset{every loop/.style={min distance=5mm,in=0,out=60,looseness=5}}

\node at (-1.2,1.0) {${H^+}$};

\node[draw, circle] (n1) at (180:1.1) {$1$};
\node[draw, circle] (n2) at (0:1.1) {$2$};


\draw[<-] (n2) edge    (n1);
\path[->] (n1) edge [in=110,out=70,looseness=6] (n1);
\path[->] (n2) edge [in=110,out=70,looseness=6] (n2);

\end{tikzpicture}
\caption{Communication graphs $G$, $H^+$, and $H^-$ used in the proof of Theorem~\ref{thm:contractrate} for $n=2$}
\label{fig:thm:lower:bound:markov:graphs}
\end{figure}

\bigskip

\noindent
Assume now that $n\geq 3$.
Let us consider the following two communication graphs depicted in Figure~\ref{fig:thm:lower:bound:markov:graphs:general}:
\begin{itemize}
\item  Graph $K$:  nodes~$2,3,\dots,n$ are fully connected and
	there is the additional edge $(2,1)$
\item Graph $L$: nodes~$2,3, \dots,n$ are fully connected and
	there are the additional edges $(p,1)$ for all nodes $p\in \{3,4,\dots,n\}$
\end{itemize}

We easily check that both $K$ and $L$ are nonsplit.
We
  consider two executions of ${\cal A}$:
The first with the communication graph~$K$ in  the first round and  the initial configuration $x_1(0)= x_2(0) =0$, and $x_p(0)=1$ for $p\in\{3,4,\dots,n\}$.
The second with the communication graph~$L$ in  the first round and  the initial configuration $x_1'(0)=1$, $x_2'(0) =0$, and $x_p'(0)=1$ for $p\in\{3,4,\dots,n\}$.
Because of the definition of an averaging algorithm, in the first execution with the communication graph~$K$, we have 
	$$x_1(1) = 0  \    \mbox{ and } \     x_2(1) = a \in [ 0, 1] \, .$$
Likewise, in the second one with~$L$, we have 
	$$x_1'(1) =1 \    \mbox{ and } \   x_2'(1) = b \in [ 0 , 1]   \, . $$
From the fact that $K \sim_2 L$, $x_2(0)=x_2'(0)$, and $x_p(0)=x_p'(0)$ for all incoming neighbors~$p$ of process~$2$ in either of the two graphs, we deduce that $a = b$.
Since ${\cal A}$ is $\alpha$-contracting,  it follows that
	$$ a \leq \alpha  \    \mbox{ and } \    a \geq 1- \alpha  \, . $$
But this is only possible if  $\alpha \geq 1/2$, which concludes the proof also in the case $n\geq3$.
\end{proof}

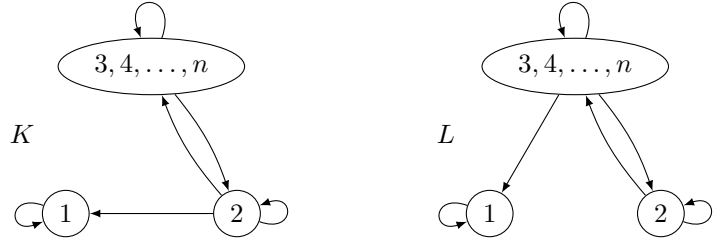
\begin{figure}
\centering
\begin{tikzpicture}[>=latex]
\tikzset{every loop/.style={min distance=5mm,in=0,out=60,looseness=5}}

\node at (-1.7,0.4) {${K}$};

\node[draw, circle] (n1) at (210:1.3) {$1$};
\path[->] (n1) edge [in=200,out=160,looseness=6] (n1);
\node[draw, circle] (n2) at (330:1.3) {$2$};
\path[->] (n2) edge [in=20,out=340,looseness=6] (n2);
\node[draw, ellipse, draw] (n3) at (90:1.3) {$3,4,\dots,n$};
\path[->] (n3) edge [in=110,out=70,looseness=6] (n3);

\draw[->] (n2) edge   [bend left=10] (n3);
\draw[<-] (n2) edge   [bend right=10] (n3);


\draw[->] (n2) edge    (n1);


\end{tikzpicture}
\hspace{1.5cm}
\begin{tikzpicture}[>=latex]
\tikzset{every loop/.style={min distance=5mm,in=0,out=60,looseness=5}}

\node at (-1.7,0.4) {${L}$};

\node[draw, circle] (n1) at (210:1.3) {$1$};
\path[->] (n1) edge [in=200,out=160,looseness=6] (n1);
\node[draw, circle] (n2) at (330:1.3) {$2$};
\path[->] (n2) edge [in=20,out=340,looseness=6] (n2);
\node[draw, ellipse, draw] (n3) at (90:1.3) {$3,4,\dots,n$};
\path[->] (n3) edge [in=110,out=70,looseness=6] (n3);

\draw[->] (n2) edge   [bend left=10] (n3);
\draw[<-] (n2) edge   [bend right=10] (n3);


\draw[->] (n3) edge    (n1);

\end{tikzpicture}
\caption{Communication graphs $K$ and $L$ used in the proof of Theorem~\ref{thm:contractrate} for the case $n\geq 3$}
\label{fig:thm:lower:bound:markov:graphs:general}
\end{figure}

\begin{algorithm}[h]
\small
\begin{algorithmic}[1]
\REQUIRE{}
  \STATE $x_p \in [0,1]$
\ENSURE{}
  \STATE send $x_p$ to other process and receive $x_q$ if $(q,p)\in E(k)$
  \IF{$x_q$ was received}
    \STATE $x_p \gets x_p/3 + 2x_q/3$
  \ENDIF
\end{algorithmic}
\caption{Averaging algorithm for two processes with contraction rate $1/3$}
\label{algo:2:procs}
\end{algorithm}

Both lower bounds in Theorem~\ref{thm:contractrate} are tight.
Indeed  the {\em midpoint algorithm} that we will introduce in Section~\ref{sec:mid}
	is $(1/2)$-safe, and so $(1/2)$-contracting in the network model of nonsplit communication
	graphs.
For a two process network,  Algorithm~\ref{algo:2:procs} has a contraction rate of $1/3$:
	 we easily check that for every positive integer~$k$,
	$$ \lvert x_1(k) - x_2(k)\rvert = \frac{\lvert x_1(k-1) - x_2(k-1) \rvert} {3}$$
	be the communication graph at round~$k$ equal to $H^+$, $H^-$, or $G$.

\section{Amortized averaging algorithms }\label{sec:amortized}

Proposition~\ref{prop:productrooted} and Theorem~\ref{thm:nonsplit} suggest to consider a new type 
	of distributed algorithms, that we call  {\em amortized averaging algorithms}, in which each process repeatedly
	first collects values during $n-1$ rounds, and then computes a weighted  average 
	of values it has received every $n-1$ rounds.
Roughly speaking, an amortized averaging algorithm is an averaging algorithm with the granularity of macro-rounds 
	consisting in  blocks of $n-1$ consecutive rounds.

We now make this notion 	more precise.
First let us fix some notation.
Macro-round $\ell$ is the sequence of rounds $(\ell-1)(n-1) +1, \ldots, 	\ell (n-1)$.
We consider algorithms for which every variable $x_p$ is updated only at the end of macro-rounds;
	$x_p(\ell)$ will denote the value of $x_p$ at  the end of round~$\ell (n-1)$, as no confusion can arise.
Given some communication pattern~$(G(k))_{k\ge 1}$,   the communication graph at macro-round~$\ell$
	is equal to:
	$$ \hat{G}(\ell) = G((\ell-1)(n-1) +1)  \circ \ldots \circ G(	\ell (n-1)) \, .$$ 
	
Each process $p$ can record the set of values it has received during macro-round~$\ell$, 
	namely the set $V_p(\ell) = \{ x_q(\ell -1) \mid q \in \In_p	\big( \hat{G}(\ell) \big) \}$,
	but in anonymous networks, $p$ cannot determine the set of its incoming neighbors
	in $\hat{G}(\ell)$.
This is in contrast to networks with unique process identifiers where each process~$p$ can determine 
	the membership of  $ \In_p \big( \hat{G}(\ell) \big) $ 
	by piggybacking the name of sender onto every message, and so can compute 
	the set
	$W_p(\ell) = \{ \big (q, x_q(\ell -1) \big ) \mid q \in \In_p	\big( \hat{G}(\ell) \big) \}$.
In particular, each process can determine the 
	multiset of values that it has received during a macro-round, counted with their multiplicities.

In consequence in any anonymous network, we can define the {\em  amortized version of} 
	an averaging algorithm ${\cal A}$ with weights in update rules that depend only on the sets of received 
	values: at the end of every macro-round~$\ell$, each process~$p$  adopts a new value by applying
	the same update rule as in ${\cal A}$ with the macro-set $V_p(\ell)$.
Based on the above discussion, this definition can be extended to averaging algorithms 
	with update rules involving the sets of incoming neighbors when processes have 
	unique identifiers.
For instance, the amortized version of the mean-value algorithm is defined in any anonymous network while
	the amortized equal-neighbor algorithm requires to have unique process identifiers. 
	
In both cases, the new value adopted by process~$p$ at the end of macro-round~$\ell$ lies within the interval
 	formed by the values of its incoming neighbors in the communication graph~$\hat{G}(\ell)$:
	 if  $\hat{m}_p(\ell-1)$  is the minimum of the values in $V_p(\ell)$ 
	 and $\hat{M}_p(\ell-1)$ is the maximum, then
	$$  \hat{m}_p(\ell -1) \leq x_p(\ell ) \leq   \hat{M}_p(\ell-1)  \, .$$

Combining Proposition~\ref{prop:productrooted} and Theorem~\ref{thm:nonsplit}, 
	we immediately derive the following central result for amortized averaging algorithms.
	
\begin{thm}\label{thm:amortizedsafe}
In any rooted network model, the amortized version of a $\varrho$-safe averaging algorithm 
	solves approximate consensus and achieves $\varepsilon$-agreement in  
	$(n-1) \big\lceil \log_{\frac{1}{1-\varrho}}  \big( \frac{1}{\varepsilon} \big) \big\rceil$ rounds.
\end{thm}

In order to fix decision times, we have assumed from the beginning that each process
	knows the number~$n$ of processes or at least an upper bound on~$n$.
However, regarding the {\em asymptotic consensus} problem -- obtained by substituting limit values
	for decision values in the specification of approximate consensus --, this global knowledge on~$n$ 
	is actually useless  for averaging algorithms.
	
In contrast,  update rules in amortized averaging algorithms require that the
	number of processes in the network is known to all processes.
In other words,  amortized averaging algorithms are defined only under 
	the assumption of this global knowledge in the network, 
	even for solving asymptotic consensus.
In fact, we can immediately adapt the definition of amortized averaging algorithms to the
	case where~$n$ is a fixed parameter and then easily verify that Theorem~\ref{thm:amortizedsafe}
	still holds when $n$ is only an upper bound on 	 the exact number of processes.

\subsection{A quadratic-time algorithm in rooted networks with process identifiers}

In~\cite{CBFN15}, we proposed an approximate consensus 
	 algorithm with a quadratic decision time. 
The algorithm (cf.\ Algorithm~\ref{algo:translation}) 
	does not work in anonymous networks as it uses process identifiers 
	so that each process can determine whether some value that it has received several times
	during a macro-round is originated  from the same process or not.

\begin{algorithm}[h]
\small
\begin{algorithmic}[1]
\REQUIRE{}
  \STATE $x_p \in [0,1]$ and $W_p \leftarrow \{(p,x_p)\}$
\ENSURE{}
  \STATE send $W_p$ to all processes in $\Out_p(k)$ and receive $W_q$ from all processes $q$ in $\In_p(k)$
  \STATE $W_p \leftarrow \bigcup_{q \in \In_p(k)} W_q$
  \IF{$k \equiv 0 \mod n-1$}
    \STATE $x_p \leftarrow \frac{1}{|W_p|}  \sum_{(q, v) \in W_p} v$
    \STATE $W_p \leftarrow \{(p,x_p)\}$
  \ENDIF
\end{algorithmic}
\caption{Amortized equal-neighbor algorithm}
\label{algo:translation}
\end{algorithm}

The update rule (line 5) rewrites as:
 	 	$$ x_p(\ell) =  \frac{ 1  }{\big | \In_p \big( \hat{G}(\ell) \big)\} \big |}  \sum_{q \in \In_p \big( \hat{G}(\ell) \big) }  \, x_q(\ell -1) \, .$$
Hence Algorithm~\ref{algo:translation} is the amortized version of the equal-neighbor 
	algorithm.
	
From Proposition~\ref{prop:EN+MV} and Theorem~\ref{thm:amortizedsafe}, it follows that  Algorithm~\ref{algo:translation} 
	solves approximate consensus.
Moreover, because of the inequality $ \log(1-a )\leq -a $  when 
	$0 \leq a$  and
	because $ \delta \left( x(0 )\right) \leq 1$, it follows that 
		if $\ell  \geq n \log \frac{1}{\varepsilon}$,
		then $ \delta \left( x(\ell)\right) \leq\varepsilon$. 
Hence Algorithm~\ref{algo:translation}  achieves $\varepsilon$-agreement in 
	$ O \left ( n ^{2} \log \frac{1}{\varepsilon} \right ) $ rounds.

\subsection{A quadratic-time algorithm in anonymous rooted networks}

As we will show below, Algorithm~\ref{algo:translation} still works when 
	processes just collect  values without taking into account processes from which they originate.

\begin{algorithm}[h]
\small
\begin{algorithmic}[1]
\INITIALLY{}
  \STATE $x_p \in [0,1]$
  \STATE $V_p \subseteq V $, initially $ \emptyset $ 
\ROUND{}
  \STATE send $V_p$ to all processes in $\Out_p(k)$ and receive $V_q$ from all processes $q$ in $\In_p(k)$
  \STATE $V_p \leftarrow \bigcup_{q \in \In_p(k)} V_q$
  \IF{$k \equiv 0 \mod n-1$}
     \STATE $x_p \leftarrow \frac{1}{|V_p|} \sum_{v \in V_p}  v  $ 
    \STATE $V_p \leftarrow \emptyset $
  \ENDIF
\end{algorithmic}
\caption{Amortized mean-value algorithm}
\label{algo:quadratic}
\end{algorithm}

At each macro-round~$\ell$, the update rule in the resulting algorithm for anonymous networks 
	(cf.~Algorithm~\ref{algo:quadratic}, line 6) coincides with the update rule in the mean-value algorithm.
In other words, Algorithm~\ref{algo:quadratic} is the amortized version of the mean-value algorithm.	
Since the latter algorithm is a $(1/n)$-safe averaging algorithm (Proposition~\ref{prop:EN+MV}), 
	we may apply Theorem~\ref{thm:amortizedsafe} and obtain the following result.

\begin{thm}\label{thm:amormv}
In a rooted network model of $n$ processes,  the amortized mean-value algorithm solves
	approximate consensus and achieves $\varepsilon$-agreement
	in $O(n^2  \log  \frac{1}{\varepsilon} ) $ rounds.
\end{thm}

\subsection{A linear-time algorithm in anonymous rooted networks}\label{sec:mid}

We improve the above quadratic upper bound on decision times with a linear amortized averaging algorithm
	which differs from our previous algorithm in the update rule: each process adopts the mid-point of 
	the range of values it has received  during a macro-round (cf.\ Algorithm~\ref{algo:mid}).
	
\begin{algorithm}[h]
\small
\begin{algorithmic}[1]
\INITIALLY{}
  \STATE $x_p \in [0,1]$
  \STATE $m_p \in [0,1]$, initially $x_p$ 
  \STATE $M_p \in [0,1]$, initially $x_p$ 
\ROUND{}
  \STATE send $(m_p,M_p)$ to all processes in $\Out_p(k)$ and receive $(m_q,M_q)$ from all processes $q$ in $\In_p(k)$
  \STATE $m_p \gets \min\big\{ m_q \mid q\in \In_p(k)\big\}$
  \STATE $M_p \gets \max\big\{ M_q \mid q\in \In_p(k)\big\}$
  \IF{$k \equiv 0 \mod n-1$}
    \STATE $x_p \gets (m_p+M_p)/2$
    \STATE $m_p \gets x_p$
    \STATE $M_p \gets x_p$
  \ENDIF
\end{algorithmic}
\caption{Amortized mid-point algorithm}
\label{algo:mid}
\end{algorithm}

Let us now consider  the {\em mid-point algorithm} that is the simple averaging algorithm
	with the mid-point update rule.
In other words,  Algorithm~\ref{algo:mid} is the amortized version of  the mid-point algorithm.

By definition, the mid-point algorithm is $(1/2)$-safe.
By Theorem~\ref{thm:nonsplit}, it follows that the mid-point algorithm is 
	a $(1/2)$-contracting  approximate consensus algorithm in any nonsplit network
	model, with a constant decision time.
This improves the linear time complexity of the equal neighbor averaging algorithm~\cite{CBFN15} 
	for this type of network models and demonstrates that  the $1/2$ lower bound 
	in Theorem~\ref{thm:contractrate} for networks with $n\geq 3$ 
	processes is tight.
                                                           
Furthermore Theorem~\ref{thm:amortizedsafe} implies that  Algorithm~\ref{algo:mid} solves approximate 
	consensus  in any rooted network model with a decision time in $O(n  \log  \frac{1}{\varepsilon} ) $ rounds.

\begin{thm}\label{thm:amormp}
In a rooted network model of $n$ processes, 
	the amortized mid-point algorithm solves approximate consensus and achieves $\varepsilon$-agreement 
	in $(n-1)\big\lceil \log_2\frac{1}{\varepsilon} \big\rceil$ rounds.
\end{thm}

Hence the amortized mean-value algorithm and  the amortized mid-point algorithm both work in any anonymous 
	network model, under the assumption that the number of processes is known to all processes.
However while the first algorithm has quadratic decision times and requires each process to store $n$ values
	per round and to transmit $n$ values per message,
	the amortized mid-point algorithm solves approximate consensus in linear-time and with only a constant number
	of (namely, two)  values per  process and per message.
A similar result has been recently obtained by Olshevsky~\cite{Ols15} with a linear-time algorithm where
	each process maintains two variables.
Unfortunately, this algorithm works with a fixed communication graph that further ought to be 
	bidirectional and connected.

\section{Resiliency of amortized averaging algorithms}\label{sec:resiliency}

%

In this section, we discuss the resiliency of our amortized averaging algorithms against 
	a wrong estimate of the number of processes or against a partial failure of the assumption that communication graphs
	are permanently rooted.

Firstly consider some communication pattern  in which only part of communication graphs are rooted:
	suppose that $N-1$ communication graphs in any macro-round of $n-1$ consecutive rounds are guaranteed to be rooted.
Then there are at least $n-1$ rooted communication graphs  in any sequence of 
	$ L=  \left \lceil \frac{n-1}{N-1} \right\rceil$ macro-rounds of length $n-1$,
	and every product of~$L$  communication graphs 
	of macro-rounds is nonsplit.

Secondly suppose that the network model is indeed rooted but processes do not
know the exact number~$n$
	of processes and only knows an estimate~$N$ on $n$.
Macro-rounds in the amortized averaging algorithms then consist of $N-1$ rounds (instead of $n-1$),
	and so the cumulative communication graphs in such macro-rounds may be not nonsplit
	when $N < n$.
However the communication graph over any block of $L= \left \lceil \frac{n-1}{N-1} \right\rceil$ 
	macro-rounds of length $N-1$ is nonsplit.

In both cases, the above discussion leads to introduce the notion of  $K$-{\em nonsplit\/} network model defined as 
	any network model~${\cal N}$ such that every product of $K$  communication graphs 
	from~${\cal N}$ is nonsplit.
Theorem~\ref{thm:nonsplit} can then be extended as follows:

\begin{thm}\label{thm:Knonsplit}
In a $K$-nonsplit network model, a $\varrho$-safe averaging algorithm is $(1-\varrho^K)$-contracting
	over  each block of $K$ consecutive rounds, i.e., for every non negative integer~$k$, we have
	$$ \delta \big( x(k+K) \big) \leq ( 1-\varrho^K) \, \delta \big( x(k) \big) \, .  $$
\end{thm}

\begin{proof}
Let $p$ and $q$ two processes such that $ \delta \big( x(k+K) \big) = x_p(k+K)  - x_q(k+K)  $.
Because the communication graph $G(k+1)\circ \dots \circ G(k+K)$ is nonsplit, there exist $2K+1$ processes 
	$p_1, \dots, p_K, q_1, \dots, q_K, r$ such that 
	$$\left \{ \begin{array}{ll}
	p_K = p, \ q_K=q,  & \\
	p_{i-1} \in \In_{p_i} \big(G(k+i) \big), \  q_{i-1} \in \In_{q_i} \big(G(k+i) \big), & \mbox{for each integer } i, \ \ 2\leq i \leq K \\
	r \in  \In_{p_1} \big(G(k+1) \big) \cap \In_{q_1} \big(G(k+1) \big)  \, . & 
	\end{array} \right.$$

Let 	$p_K^+$ and $p_K^-$ denote two processes such that 
	$$x_{p_K^+ }(k+ K-1) = M_{p_K} (k + K-1) \mbox{ and }  x_{p_K^- }(k + K-1) = m_{p_K}  (k + K-1) \, .$$
Similarly we introduce two processes $q_K^+$ and $q_K^-$.
Because the algorithm is $\varrho$-safe, we have
\begin{equation*}
\begin{split} 
x_{p_K }(k+ \!K)  - \! x_{q_K }(k+ \!K) &  \leq  (1- \!\varrho) \big (x_{p_K^+ }(k+ \!K- \!1)  - \! x_{q_K^- }(k+ \!K- \! 1) \big)  \\ & \qquad  + \!
                                                                      \varrho \big (x_{p_K^- }(k+ \!K-1)  - \! x_{q_K^+ }(k+ \!K-1) \big) \\ 
 & \leq  (1- \! \varrho) \big (M(k) - m(k) \big)  +  \varrho \big (x_{p_{K-1}}( k+ K- \! 1)  -  x_{q_{K-1} }(k+ K- \! 1) \big) . 
\end{split}
\end{equation*}
By iterating the same argument $K$ times, we obtain
\begin{align*} 
x_{p_K }(k+ \!K)  - \! x_{q_K }(k+ \!K)  
 & \leq (1- \! \varrho) (1 + \varrho+ \dots + \varrho^{K-1}) \big (M(k) - m(k) \big)  +  \varrho^K \big (x_r( k)  -  x_r (k) \big) . 
\end{align*}
Hence we have $$    \delta \big( x(k+K) \big) \leq   (1- \! \varrho^K) \delta \big( x(k) \big) \, ,$$
which concludes the proof.
\end{proof}

As an immediate corollary of Theorem~\ref{thm:Knonsplit}, we obtain that the amortized version of a 
	$\varrho$-safe averaging algorithm is resilient against a wrong estimate of the number of processes 
	or against a partial failure of the assumption of  a rooted network model, and its decision times are 
	multiplied by a factor in $O \big( L \varrho^{-L} \big)$.
Note that the theorem measures the number of macro-rounds, and not the number of rounds.

\begin{thm}\label{thm:resiliency}
The amortized version of a $\varrho$-safe averaging algorithm solves approximate consensus 
	even with an erroneous number of processes or with a communication pattern where only part of 
	communication graphs are rooted.
Moreover it achieves $\varepsilon$-agreement in 
	$O\left( L   \varrho^{-L} \cdot \log  \big( \frac{1}{\varepsilon} \big)  \right)$ 
	macro-rounds if $N$ denotes the estimate on process number or if  only $N-1$ rounds in each block of
	$n-1$ consecutive rounds are guaranteed to have a rooted communication graph and where $L = \left\lceil\frac{n-1}{N-1}\right\rceil$.
\end{thm}


\section{Rounding}\label{sec:rounding}

In this section, we take into account the additional constraint that processes can only store
	and transmit quantized values.
This model provides a good approximation for networks with storage constraints or
	with finite bandwidth channels.
	
We include this constraint by quantizing each averaging update rule.
For that, we fix some positive integer~$Q$ and choose a rounding function, denoted  $[\,  . \,  ]$,
	which rounds down (or rounds up) to the nearest multiple of $1/Q$. 
Then the quantized update rule for process~$p$ at round~$k$ writes	
	\begin{equation*}\label{eq:roundingupdate}
	x_p(k) = \left [    \sum_{q \in \In_p(k)} w_{qp}(k) \, x_q(k-1)   \right ] 
	\end{equation*}
where the $w_{qp}(k)$ denote the weights in the average of the original algorithm.
Besides we assume that all initial values are multiples of $1/Q$.	

\subsection{Quantization and  mid-point}

Nedi\'c et al.~\cite{NOOT09} proved that in any strongly connected network model, every 
	quantized averaging algorithm  with the update rule~(\ref{eq:roundingupdate}) solve 
	exact consensus (and so approximate consensus). 
Because of the impossibility result for exact consensus~\cite{CBFN15}, their result does not
	hold if the strong connectivity assumption is weakened into the one of rooted network models. 
	
For the same reason, if $\varepsilon < 1/Q$, then $\varepsilon$-consensus cannot be generally 
	achieved in a rooted network model by a quantized averaging algorithm or  its amortized 
	version.
In this section, we prove that the quantization of the amortized mid-point algorithm 
	indeed achieves $1/Q$-agreement in any rooted network model. 
In case one is interested in obtaining a (suboptimal) precision $\varepsilon > 2/Q$, a small adaptation of the proof
        yields that this is already achieved earlier on in
        round $(n-1)\left( \big\lfloor \log_2 \frac{Q-2}{Q\varepsilon - 2} \big\rfloor +1 \right)$.

\begin{thm}\label{thm:roundingmid}
In a rooted network model, quantization of the amortized mid-point algorithm achieves
	$1/Q$-agreement by round~$ (n-1)\left( \big\lfloor \log_2 (Q-2) \big\rfloor +2 \right)$. 
Moreover	in every execution, the sequence of values of every process $p$ converges to a limit $x_p^*$ that is
	a multiple of $1/Q$ in finite time, and
	for every pair of processes $p,\, q$, we have either $x_p^* = x_q^* $ or $\big|x_p^* - x_q^*\big|=1/Q $.
\end{thm}
 
\begin{proof}
 We adopt the same notation as in the previous sections; 
 	in particular, we set $m (\ell) = \min_{p \in [n]} \big ( x_p(\ell) \big) $
	and $M (\ell ) = \max_{p \in [n]} \big ( x_p(\ell) \big)$.
Moreover we give the proof in the case of the rounding down rule.
	
Since  the mid-point algorithm  is $(1/2)$-safe in any nonsplit network model,	
	for every macro-round $\ell$ we have
	$$ \delta \left( x(\ell) \right) \leq 
		\frac{ \delta \left( x (\ell-1) \right)}{2}   + \frac{1}{Q}   \, .$$
Hence we obtain 
	\begin{equation}\label{eq:quantized}
	 \delta \left( x (\ell) \right) \leq  \frac{ \delta \left( x (0) \right) }{2^{\ell}}  + \frac{1}{Q} \,
	                                                 \left(1+ \frac{1}{2}+ \dots + \frac{1}{2^{\ell -1}} \right ) 
	                                                 \leq  \frac{ 1 }{2^{\ell}}  + \frac{2}{Q} \, \left( 1-  \frac{ 1 }{2^{\ell}}   \right) \, .
	 \end{equation}
It follows that $ \delta \left( x (\ell) \right) < 	\frac{3}{Q} $ for every macro-round $\ell$ with $\ell >  \log_2(Q-2)$.

Let $L= \big\lfloor \log_2 (Q-2) \big\rfloor +1$; we have  $ \delta \left( x (L) \right) < \frac{3}{Q} $, 
	and so at the end of macro-round $L$, processes adopt at most three different values which are multiples of $1/Q$
	in an interval of the form $ [ (j-1)/Q , (j+1)/Q ]$.
We now consider the following three cases:
\begin{enumerate}
\item $M (L) = m (L)$;  then all the $x_p$'s remains equal to $M(L)$ from macro-round~$L$.
The theorem trivially holds in this case.
\item $M (L) = m (L) + 1/Q$.
	By an inductive argument, we see that for each process~$p$,
	 the sequence $\big( x_p(\ell) \big)_{\ell \geq 1}$ is eventually constant\footnote{%
	Note that there is no bound on the convergence time.}
	 and its limit is either $m(L)$ or $m(L)+1/Q$.
This shows that the theorem holds in this case.
\item Otherwise $M (L) = m (L) + 2/Q$.
We show that $M(L+1)-m(L+1) \leq 1/Q$.

Let ${\cal M}$ denote the set of  processes with values  $M(L)$ at the end of macro-round~$L$.
Then for every process~$p$, we have
	$$ x_p(L+1) = \left\{ \begin{array}{ll}
	                                  m(L) \mbox{ or } m(L) + 1/Q  & \mbox{ if } p \notin {\cal M} \\
	                                  m(L) + 1/Q  \mbox{ or } m(L) + 2/Q  & \mbox{ if } p \in {\cal M} 
	                                  \end{array} \right. $$
Suppose for contradiction that $M(L+1) = m(L+1) + 2/Q$.
Then we have $ m(L+1) = m(L) $, $ M(L+1) = M(L) = m(L) + 2/Q$, and  there exists some process $q_0  \in {\cal M}$ 
	such that $$x_{q_0}(L+1) = x_{q_0}(L) = m(L) + 2/Q \, . $$
In the communication graph $\hat{G}(L+1)$ of macro-round $L+1$,
	all of~$q_0$'s incoming neighbors are in~${\cal M}$.
Combined with the fact that $\hat{G}(L+1)$  is nonsplit, this implies that every process~$p$ 
	has an incoming neighbor in ${\cal M}$.
Hence for every process~$p\notin {\cal M}$, we have $$x_p(L+1)=  m(L) + 1/Q \, , $$
	leading to a contraction with the equality $ m(L+1) = m(L) $.
	
It follows that either $M(L+1) = m(L+1)$ or $M(L+1) = m(L+1) + 1/Q$.
In other words, either case (1) or case (2) occurs at the end of macro-round $L+1$.
\qedhere
\end{enumerate}
\end{proof}

Inequality (\ref{eq:quantized}) can be easily generalized for the amortized version of any 
	$\varrho$-safe algorithm.
However the rest of the above proof works only for value $1/2$ of parameter~$\varrho$.
Hence Theorem~\ref{thm:roundingmid} seems to be specific to the amortized mid-point algorithm.

 \subsection{Approximate consensus versus 2-set consensus}
 
 We now consider the {\em 2-set consensus problem} which is another natural generalization of 
 	the consensus problem.
Instead of requiring that processes agree to within any positive real-valued tolerance $\varepsilon$,
	processes have to decide on at most 2 different values:
\begin{description}
\item{\em Agreement.}	There are at most two  different decision values.
\end{description}

Formally, each process starts with an input value from the set~$V$ of multiples of $1/Q$ 
	and has to output a decision value from $V$ in such a way that the termination and validity conditions in
	the consensus specification as well as  the above agreement condition are satisfied.
	
The 2-set consensus problem naturally reduces to approximate consensus: processes round off  their 
	$1/Q$-agreement output values.
Unfortunately, the use of averaging procedures to solve approximate consensus leads processes to exchange values out of the set~$V$ 
	in the resulting 2-set consensus algorithms.
The quantized amortized mid-point algorithm allows us to overcome this problem, and Theorem~\ref{thm:roundingmid}
	shows that in any rooted network model,  this algorithm achieves 2-set consensus in~$ (n-1)\left( \big\lfloor \log_2 (Q-2) \big\rfloor +2 \right)$
	rounds.

The above discussion shows that the 2-set consensus problem is  solvable in a dynamic network model if 
	all the communication graphs are rooted.
In particular, 2-set consensus is solvable in a asynchronous complete network with a minority of faulty senders since
	nonsplit rounds  can be implemented in this model.
Combined with the impossibility result in~\cite{dPMR99}  	in the case of a strict majority of faulty processes,
	we obtain an exact characterization of  the sender faulty models for which 2-set consensus is solvable 
	in asynchronous systems if the number of processes $n$ is odd 
	and a small gap (namely $n/2$ faulty processes) when $n$ is even.
	
Our positive result can be interestingly compared with  the 2 faulty processes boundary~\cite{BG93,HS93,SZ93} between
	possibility and impossibility of the original (and stronger)  2-set consensus problem~\cite{Cha93} where decision values
	ought to be initial values, instead of being in the range of the initial values.
That points out the crucial role of the validity condition on the solvability of 2-set consensus.

\bibliographystyle{plain}
\bibliography{agents}


%

\end{document}